\documentclass[11pt, letterpaper]{amsart}
\usepackage{graphicx, amssymb}

\newtheorem{thm}{Theorem}[section]

\newtheorem{prop}[thm]{Proposition}
\theoremstyle{definition}
\newtheorem{defn}[thm]{Definition}

\theoremstyle{remark}
\newtheorem{rem}[thm]{Remark}

\numberwithin{equation}{section}

\newcommand{\set}[1]{\left\{#1\right\}}
\newcommand{\Real}{\mathbb R}
\newcommand{\Natural}{\mathbb N}

\newcommand{\such}{\ | \ }
\newcommand{\prob}{\mathbb{P}}
\newcommand{\Exp}{\mathcal E}
\newcommand{\parti}{\mathbb{T}}

\newcommand{\qprob}{\mathbb{Q}}
\newcommand{\expec}{\mathbb{E}}

\newcommand{\basis}{(\Omega,  \, (\F_t)_{t \in \Real_+}, \, \prob)}

\newcommand{\F}{\mathcal{F}}

\newcommand{\cadlag}{c\`adl\`ag}
\newcommand{\caglad}{c\`agl\`ad}
\newcommand{\ud}{\mathrm d}
\newcommand{\inner}[2]{\left \langle #1 , #2 \right \rangle}

\newcommand{\naone}{NA$1_{\mathsf{s}}$}
\newcommand{\naonee}{\emph{NA}$1_{\mathsf{s}}$}

\newcommand{\X}{\mathcal{X}_{\mathsf{s}}}

\newcommand{\simplex}{\overline{\Delta}^d}

\newcommand{\co}{\mathsf{c}}

\newcommand{\pare}[1]{\left(#1\right)}
\newcommand{\bra}[1]{\left[#1\right]}
\newcommand{\dbra}[1]{[\kern-0.15em[ #1 ]\kern-0.15em]}
\newcommand{\dbraco}[1]{[\kern-0.15em[ #1 [\kern-0.15em[}
\newcommand{\dbraoc}[1]{]\kern-0.15em] #1 ]\kern-0.15em]}

\newcommand{\indic}{\mathbb{I}}

\newcommand{\ucp}{\mathsf{uc} \prob}

\newcommand{\us}{u_{\mathsf{s}}}
\newcommand{\ou}{u}
\newcommand{\oX}{\mathcal{X}}

\newcommand{\ox}{\widehat{X}}

\newcommand{\dfn}{\, := \, }

\begin{document}

\title[Approximation of no-short-sales wealth processes via simple trading]{Multiplicative approximation of wealth processes involving no-short-sales strategies via simple trading}%
\author{Constantinos Kardaras}%
\address{Constantinos Kardaras, Mathematics and Statistics Department, Boston University, 111 Cummington Street, Boston, MA 02215, USA.}%
\email{kardaras@bu.edu}%
\author{Eckhard Platen}%
\address{Eckhard Platen, School of Finance and Economics \& Department of Mathematical Sciences, University of Technology, Sydney, P.O. Box 123, Broadway, NSW 2007, Australia.}%
\email{eckhard.platen@uts.edu.au}%

\thanks{The first author would like to thank the warm hospitality of the School of Finance and Economics of the University of Technology Sydney, where the major part of this work was carried out.}%
\subjclass{60H05 · 60H30 · 91B28}%
\keywords{Semimartingales; buy-and-hold strategies; stochastic integral; arbitrages of the first kind; utility maximization}%

\date{\today}%
\begin{abstract}
A financial market model with general semimartingale asset-price processes and where agents can only trade using no-short-sales strategies is considered. We show that wealth processes using continuous trading can be approximated very closely by wealth processes using simple combinations of buy-and-hold trading. This approximation is based on controlling the \emph{proportions} of wealth invested in the assets. As an application, the utility maximization problem is considered and it is shown that optimal expected utilities and wealth processes resulting from continuous trading can be approximated arbitrarily well by the use of simple combinations of buy-and-hold strategies.
\end{abstract}

\maketitle

\setcounter{section}{-1}

\section{Introduction}

In frictionless financial market modeling, semimartingale discounted asset-price processes are ubiquitous. On one hand, this structure is enforced by natural market viability conditions --- see for example \cite{MR1304434} and \cite{KarPla08a}. On the other hand, the powerful tool of stochastic integration with respect to general predictable integrands already permits answers to fundamental economic questions, as is for example the classical \emph{utility maximization problem} --- see \cite{MR1722287, MR2023886} for a very general framework.

In financial terms, stochastic integration using \emph{general} predictable integrands translates  into allowing for \emph{continuous} trading in the market. Its theoretical importance notwithstanding, since it allows for existence and elegant representations of optimal wealth processes, continuous trading is but an ideal approximation. In reality, agents in the market can only use simple finite combinations of buy-and-hold strategies. It is therefore natural to question the practical usefulness of such modeling approach. Furthermore, in the context of numerical approximations, where time-discretization is inevitable, computer modeling of hedges can simulate only simple buy-and-hold trading.

\smallskip

The questions we are dealing with in the present paper are the following: \emph{Can wealth processes that are obtained by allowing continuous trading be closely approximated via simple buy-and-hold trading? If the answer to the previous question is affirmative, how can this eventually be achieved?}

Our contribution is an approximation result for wealth processes involving no-short-sales strategies allowing only simple wealth processes. In order to achieve this, we establish an interesting intermediate result on \emph{multiplicative}-type approximation of positive stochastic integrals. This is carried out by following a time-discretized continuous trading strategy in proportional, rather than absolute, terms. The actual number of units held in the portfolio still remains constant between trading dates; however, the investment strategy is parametrized by fractions. Not only is the former choice of discrete-time approximation more reasonable from a trading viewpoint under a range of objectives, it also ensures that the investor's self-financing wealth process stays nonnegative, therefore admissible, even in the presence of jumps in the asset-price process. Note that, in the case where jumps are involved in the market model, a use of the classical \emph{additive} approximation using the dominated convergence theorem for stochastic integrals might fail to guarantee that the approximating wealth processes are nonnegative.

We also provide an application of our approximation result to the expected utility maximization problem. Specifically, under weak economic assumptions, it is shown that the indirect utilities and (near-)optimal wealth processes under the possibility of no-short-sales continuous trading can be approximated arbitrarily well using simple combinations of buy-and-hold strategies.

\smallskip

There is a wealth of literature on approximations of stochastic integrals. In the context of financial applications, we mention for example \cite{MR1442320} dealing with continuous-path assets, as well as \cite{MR1971602}, where a result that is useful in approximating the optimal wealth process for the exponential utility maximization problem is proved. The analysis in the present paper is different, as we are interested in cases where wealth has to remain positive. To the best of the authors' knowledge, no previous work in this respect for asset-processes that include jumps has appeared before.

\bigskip

The structure of the paper is as follows: Section \ref{sec: market and trading} introduces the market model, where no-short-sales trading is allowed. Section \ref{sec: approx cont via bandh} contains the statements and proofs of the basic approximation results. Finally, Section \ref{sec: util max} contains the application to the utility maximization problem.

\section{The Financial Market and No-Short-Sales Trading} \label{sec: market and trading}

\subsection{The financial market model}

The evolution of $d$ risky assets in the market is modeled via \emph{nonnegative} and \emph{\cadlag} (right-continuous with left-hand limits) stochastic processes $S^1, \ldots, S^d$, where we write $S = (S^1_t, \ldots, S^d_t)_{t \in \Real_+}$. We assume in the sequel that all wealth processes, including the above assets, are denominated in units of another traded ``baseline'' asset; this could be, for example, the money market account. All processes are defined on a filtered probability space $\basis$. Here, $\prob$ is a probability on $(\Omega, \F_\infty)$, where $\F_\infty \dfn \bigvee _{t \in \Real_+} \F_t$, and $(\F_t)_{t \in \Real_+}$ is a filtration satisfying the usual assumptions of right-continuity and saturation by all $\prob$-null sets of $\F$. It will be assumed throughout that $\F_0$ is trivial modulo $\prob$.
%
%
%

\subsection{Trading via simple no-short-sales strategies} \label{subsec: simple no-short-sale}

In the market with the discounted asset-price processes described above, economic agents can trade in order to reallocate their wealth. Realistic trading consists of finite combinations of buy-and-hold strategies. We model this by considering processes of the form $\theta \dfn \sum_{j=1}^n \vartheta_{\tau_{j-1}} \indic_{\dbraoc{\tau_{j-1}, \tau_j}}$, where each $\tau_j$, $j = 0, \ldots, n$, is a finite stopping time with $0 = \tau_0 < \tau_1 < \ldots < \tau_n$, and where each $\vartheta^i_{\tau_{j-1}}$
is $\F_{\tau_{j-1}}$-measurable for $i=1, \ldots, d$ and $j = 1, \ldots, n$. Starting from initial capital $x \in \Real_+$ and investing according to the aforementioned simple strategy $\theta$, the agent's \textsl{discounted}, with respect to the baseline asset, \textsl{wealth process} is
\begin{equation} \label{eq: wealth process simple}
X^{x, \theta} \ = \ x + \int_0^\cdot \inner{\theta_t}{\ud S_t}
\dfn x + \sum_{j = 1}^n \sum_{i = 1}^d \vartheta^i_{\tau_{j-1}} \big( S^i_{\tau_j  \wedge \cdot} - S^i_{\tau_{j-1}  \wedge \cdot} \big),
\end{equation}
where we are using $\inner{\cdot}{\cdot}$ throughout to denote (sometimes, formally) the usual Euclidean inner product on $\Real^d$. Note that the predictable process $\theta$ is modeling the units of assets held in the portfolio, and that it is \emph{piecewise constant} over time.

The wealth process $X^{x, \theta}$ of \eqref{eq: wealth process simple} could, in principle, become negative. In real markets, economic agents sometimes face institution-based trading constraints, the most important and typical example of which is the prevention of short sales. Consider a wealth process $X^{x, \theta}$ as in \eqref{eq: wealth process simple}. In order to ensure that there are no short sales of the risky assets \emph{and} the baseline asset, we ask that
\begin{equation} \label{eq: no-short-sales}
\theta_t^i \geq 0 \text{ for all } i =1, \ldots, d, \text{ as well as } \sum_{i=1}^d \theta_t^i S_{t-}^i \leq X_{t-}^{x, \theta}, \text{ for all } t \in \Real_+,
\end{equation}
where the subscript ``$t-$'' is used to denote the left-hand limit of processes at time $t \in \Real_+$. For fixed initial wealth $x \in \Real_+$, we define the set $\X (x)$ of all \textsl{no-short-sales} wealth processes using \emph{simple} trading, which are the wealth processes $X^{x, \theta}$ given by \eqref{eq: wealth process simple} such that \eqref{eq: no-short-sales} holds.  (Note that subscripts ``$\mathsf{s}$'', like the one used in the definition of $\X(x)$ for $x \in \Real_+$, will be used throughout the paper serving as a mnemonic for ``simple''.)

\subsection{Arbitrages of the first kind} \label{subsec: NUPBR}

The market viability concept we shall now introduce is a weakened version of the \textsl{No Free Lunch with Vanishing Risk} condition of \cite{MR1304434}.

\begin{defn} \label{dfn: NUPBR}
For $T \in \Real_+$, an $\F_T$-measurable random variable $\xi$ will be called an \textsl{arbitrage of the first kind on $[0, T]$} if $\prob[\xi \geq 0] = 1$, $\prob[\xi > 0] > 0$, and \emph{for all $x > 0$ there exists $X \in \X(x)$, which may depend on $x$, such that $\prob[X_T \geq \xi] = 1$}. If, in a market where only simple, no-short-sales trading is allowed, there are \emph{no} arbitrages of the first kind on \emph{any} interval $[0, T]$, $T \in \Real_+$, we shall say that condition \naone \ holds.
\end{defn}




For economic motivation and more information on condition \naone, we refer the interested reader to \cite{KarPla08a}. The next result follows in a straightforward way from Theorem 2.3 of \cite{KarPla08a}.

\begin{thm} \label{thm: semimarts}
Assume that condition \naonee \ of Definition \ref{dfn: NUPBR} holds. Then, $S$ is a semimartingale. Further, for all $X \in \bigcup_{x \in \Real_+} \X(x)$, defining $\zeta^X := \inf \{ t \in \Real_+ \such X_{t-} = 0 \text{ or } X_t = 0\}$ to be the (first) \textsl{bankruptcy time} of $X$, we have $X_{t} = 0$ for all $t \in [ \zeta^X, \infty [$ on the event $\set{\zeta^X < \infty}$.
\end{thm}

\subsection{No-short-sales continuous trading}

If condition \naone \ is in force, Theorem \ref{thm: semimarts} implies the semimartingale property of $S$. We can therefore use general stochastic integration with respect to $S$, allowing in effect agents to change their position in the assets in a continuous fashion. This form of trading is only of theoretical interest, since it cannot be implemented in reality even if one ignores market frictions, as we do here.

Starting from initial capital $x \in \Real_+$ and investing according to some predictable and $S$-integrable strategy $\theta = (\theta^1_t, \ldots, \theta^d_t)_{t \in \Real_+}$, an agent's \textsl{discounted wealth process} is
\begin{equation} \label{eq: wealth process general}
X^{x, \theta} \dfn x + \int_0^\cdot \inner{\theta_t}{\ud S_t},
\end{equation}
where in the above definition $\int_0^\cdot \inner{\theta_t}{\ud S_t}$ denotes a \emph{vector} It\^o stochastic integral --- see \cite{MR1975582}.

For an initial wealth $x \in \Real_+$, $\oX (x)$ will denote the set of all \textsl{no-short-sales} wealth processes allowing continuous trading, that is, wealth processes $X^{x, \theta}$ given by \eqref{eq: wealth process general} such that \eqref{eq: no-short-sales} holds. As this form of continuous-time trading obviously includes as a special case the simple no-short-sales trading described in \S \ref{subsec: simple no-short-sale}, we have $\X(x) \subseteq \oX(x)$ for all $x \in \Real_+$.

Under condition \naone, the conclusion of Theorem \ref{thm: semimarts} stating that $X_{t} = 0$ for all $t \in [ \zeta^X, \infty [$ on the event $\set{\zeta^X < \infty}$, where $\zeta^X := \inf \{ t \in \Real_+ \such X_{t-} = 0 \text{ or } X_t = 0\}$, extends to all $X \in \bigcup_{x \in \Real_+} \oX(x)$. Again, this comes as a straightforward consequence of Theorem 2.3 in \cite{KarPla08a}. We shall feel free to imply this strengthened version of Theorem \ref{thm: semimarts} whenever we are referring to it.

\section{Approximation of No-Short-Sales Wealth Processes via Simple Trading} \label{sec: approx cont via bandh}

In this section, we discuss an approximation result for no-short-sales wealth processes obtained from continuous trading via simple strategies. We consider convergence of processes in probability \emph{uniformly} on compact time-sets. The notation $\ucp$-$\lim_{n \to \infty} \xi^n = \xi$ shall mean that $\prob$-$\lim_{n \to \infty} \sup_{t \in [0, T]} | \xi_t^n - \xi_t| = 0$, for all $T \in \Real_+$. Note that $\ucp$-convergence comes from a metric topology. For more information on this rather strong type of convergence, we refer to \cite{MR2273672}.

\subsection{The approximation result}
We now state the main result of this section.

\begin{thm} \label{thm: approx_with_bandh}
Assume that condition \naonee \ is valid in the market. For all $x \in \Real_+$ and $X \in \oX(x)$, there exists an $\X(x)$-valued sequence $(X^k)_{k \in \Natural}$ such that $\ucp$-$\lim_{k \to \infty} X^k = X$.
\end{thm}

The proof of Theorem \ref{thm: approx_with_bandh}, which will be given in \S \ref{subsec: proof of thm: approx with bandh}, will involve a ``multiplicative'' approximation of the stochastic integral, discussed in \S \ref{subsec: propor trading} and \S \ref{subsec: approx in mult way}, which is sensible from a trading viewpoint when dealing with nonnegative wealth processes.

\begin{rem} \label{rem: approx away from zero}
In the statement of Theorem \ref{thm: approx_with_bandh}, suppose further that there  exists some $\epsilon > 0$ such that $X \geq \epsilon$. Then, it is straightforward to see that the approximating sequence $(X^k)_{k \in \Natural}$ can be chosen in a way such that $X^k \geq \epsilon$, for all $k \in \Natural$.
\end{rem}

\subsection{Proportional trading} \label{subsec: propor trading}

Sometimes it is more useful to characterize investment in \emph{relative}, rather than \emph{absolute} terms. This means looking at the \emph{fraction} of current wealth invested in some asset rather than the \emph{number of units} of the asset held in the portfolio, as we did in \eqref{eq: wealth process simple} and \eqref{eq: wealth process general}.

Under condition \naone, the validity of Theorem \ref{thm: semimarts} allows one to consider the \textsl{total returns} process $R = (R_t^1, \ldots, R_t^d)_{t \in \Real_+}$, where $R$ satisfies $R_0 = 0$ and the system of stochastic differential equations $\ud S^i_t = S^i_{t-} \ud R^i_t$ for $i=1, \ldots, d$ and $t \in \Real_+$. In other words, $S^i = S^i_0 \Exp(R^i)$, where $\Exp$ is the \textsl{stochastic exponential} operator, see \cite{MR2273672}. It should be noted that, for $i=1, \ldots, d$, the process $R^i$ only lives in the stochastic interval $\dbraco{0 , \zeta^{S^i}}$ until the bankruptcy time $\zeta^{S^i}$ of Theorem \ref{thm: semimarts}, and that it might explode at time $\zeta^{S^i}$. However, this does not affect the validity of the conclusions below, due to the fact that, by Theorem \ref{thm: semimarts}, $S^i_{t} = 0$ holds for all $t \in [ \zeta^{S^i}, \infty [$ on $\big\{ \zeta^{S^i} < \infty \big\}$, $i=1, \ldots, d$.

Let $\simplex$ denote the closed \textsl{$d$-dimensional simplex}, i.e.,
\[
\simplex \dfn \set{ \pare{z^1, \ldots, z^d} \in \Real^d \ \Big| \ z^i \geq 0 \text{ for all } i=1, \ldots, d, \text{ and } \sum_{i=1}^d z^i \leq 1}.
\]
For any predictable, $\simplex$-valued process $\pi = (\pi^1_t, \ldots, \pi^d_t)_{t \in \Real_+}$ of investment fractions, consider the process $X^{(x, \pi)}$ defined via
\begin{equation} \label{eq: wealth_mult}
X^{(x, \pi)} \dfn x \, \Exp \left( \int_0^\cdot \inner{\pi_t}{\ud R_t} \right).
\end{equation}
Observe that we are using parentheses in the ``$(x, \pi)$'' superscript of $X$ in \eqref{eq: wealth_mult} to distinguish from a wealth process of the form $X^{x, \theta} = x + \int_0^\cdot \inner{\theta_t}{\ud S_t}$, generated by $\theta$ in an additive way.

Under condition \naone, the set of all processes $X^{(x, \pi)}$ when ranging $\pi$ over all the predictable $\simplex$-valued processes is exactly equal to $\oX(x)$. This is straightforward as soon as one notices that
\[
\set{ \zeta^{X^{(x, \pi)}} < \infty } = \bigcup_{i=1}^d \set{\zeta^{S^i} < \infty, \ \sum_{j =1}^d \pi^j_{\zeta^{S^j}} \indic_{\set{\zeta^{S^j} = \zeta^{S^i}}} = 1}.
\]

\subsection{Stochastic integral approximation in a multiplicative way} \label{subsec: approx in mult way}

Start with some adapted and \caglad \ (left continuous with right limits), therefore predictable, $\simplex$-valued process $\pi$ of investment fractions. The wealth process generated by $\pi$ in a multiplicative way starting from $x \in \Real_+$ is $X^{(x, \pi)}$, as defined in \eqref{eq: wealth_mult}. Consider now some economic agent who may only change the asset positions at times contained in $\parti = \{0 =: \tau_0 < \tau_1 < \ldots < \tau_n \}$. Wanting to approximately, but rather closely, replicate $X^{(x, \pi)}$, the agent will decide at each trading instant $\tau_{j-1}$ to rearrange the portfolio wealth in such a way as to follow with a piecewise constant number of units of the asset held until the next trading time $\tau_j$ the given investment portfolio. More precisely, the agent will rearrange wealth at time $\tau_{j-1}$, $j =1, \ldots, n$, in a way such that a proportion $\pi^i_{\tau_{j-1} +} := \lim_{t \downarrow \tau_{j-1}} \pi^i_t$ is held in the $i$th asset, $i=1, \ldots, d$; the resulting number of units is then held constant until time $\tau_j$, when a new reallocation will be made in the way previously described. Starting from initial capital $x \in \Real_+$ and following the above-described strategy, the agent's wealth remains nonnegative and is given by
\begin{equation} \label{eq: Xmult}
X^{(x, \pi; \parti)} \ := \ x \prod_{j=1}^n \left\{  1 + \sum_{i = 1}^d \pi^i_{\tau_{j-1} +} \bigg( \frac{S^i_{\tau_j\wedge \cdot} - S^i_{\tau_{j-1}\wedge \cdot}}{S^i_{\tau_{j-1}\wedge \cdot}} \bigg) \right\}.
\end{equation}
Note that, for all $i=1, \ldots, d$, $j=1, \ldots, n$ and $t \in \Real_+$, the ratio $(S^i_{\tau_j\wedge t} - S^i_{\tau_{j-1}\wedge t})/S^i_{\tau_{j-1}\wedge t}$ is assumed to be zero on the event $\{S^i_{\tau_{j-1}\wedge t} = 0\}$. Using the fact that the filtration $(\F_t)_{t \in \Real_+}$ is right-continuous, it is straightforward to see that $X^{(x, \pi; \parti)} \in \X (x)$.

Consider a sequence $(\parti^k)_{k \in \Natural}$ with $\parti^k \equiv \{ \tau_0^k < \ldots < \tau^k_{n^k} \}$ for each $k \in \Natural$, where each $\tau_j^k$, for $k \in \Natural$ and $j=0, \ldots, n^k$, is a finite stopping time. We say that $(\parti^k)_{k \in \Natural}$ \textsl{converges to the identity} if, $\prob$-a.s., $\lim_{k \to \infty} \tau^k_{n^k}= \infty$ as well as $\lim_{k \to \infty} \sup_{j=1, \ldots, n^k} |\tau^k_j - \tau^k_{j-1}| = 0$.

\begin{thm} \label{thm: multipl_approx}
Assume the validity of condition \naonee. Consider any adapted and \caglad \ $\simplex$-valued process $\pi$. If $(\parti^k)_{k \in \Natural}$ converges to the identity, then $\ucp$-$\lim_{k \to \infty} X^{(x, \pi; \parti^k)} = X^{(x, \pi)}$.
\end{thm}

\begin{proof}

Under condition \naone, and in view of Theorem \ref{thm: semimarts}, we have $\ucp$-$\lim_{\epsilon \downarrow 0} X^{(x, (1 - \epsilon) \pi)} = X^{(x, \pi)}$, as well as that, for all $k \in \Natural$, $\ucp$-$\lim_{\epsilon \downarrow 0} X^{(x, (1 - \epsilon) \pi;  \parti^k)} = X^{(x, \pi; \parti^k)}$. It follows that we might assume that $\pi$ is actually $(1 - \epsilon) \simplex$-valued, where $0 < \epsilon < 1$, which means that $X^{(x, \pi)}$, as well as $X^{(x, \pi; \parti^k)}$ for all $k \in \Natural$, remain \emph{strictly} positive. Actually, since the jumps in the returns of the wealth processes involved are bounded below by $-(1 - \epsilon)$, the wealth processes themselves are bounded away from zero in compact time-intervals, with the strictly positive bound possibly depending on the path. It then follows that $\ucp$-$\lim_{k \to \infty} X^{(x, \pi; \parti^k)} = X^{(x, \pi)}$ is equivalent to $\ucp$-$\lim_{k \to \infty} \log X^{(x, \pi; \parti^k)} = \log X^{(x, \pi)}$, which is what we shall prove below.

To ease notation in the course of the proof we shall assume that $d=1$. This is done for typographical convenience only; one can read the whole proof for the case of $d$ assets, if multiplication and division of $d$-dimensional vectors are understood in a coordinate-wise sense. Also, in order to avoid cumbersome notation, from here onwards the dot ``$\cdot$'' between two processes will denote stochastic integration and $[Y,Y]$ will denote the quadratic variation process of a semimartingale $Y$.

Proceeding with the proof, write
\begin{eqnarray} \label{eq: log-relative mult 1}
  \log \pare{\frac{X^{(x, \pi; \parti^k)}}{X^{(x, \pi)}}} &=& \sum_{j=1}^{n^k} \log \left(  1 + \pi_{\tau^k_{j-1} +} \frac{S_{\tau^k_j\wedge \cdot} - S_{\tau^k_{j-1}\wedge \cdot}}{S_{\tau^k_{j-1}\wedge \cdot}}  \right) \\
\nonumber    &-& \left( \pi \cdot R - \frac{1}{2} [\pi \cdot R^\co, \pi \cdot R^\co] -  \sum_{t \leq \cdot} \left( \pi_t \Delta R_t - \log \left( 1 + \pi_t \Delta R_t \right) \right) \right),
\end{eqnarray}
where $R^\co$ is the uniquely-defined continuous local martingale part of the semimartingale $R$. Define the adapted \caglad \  process $\eta := (\pi / S_-) \indic_{\{S_- > 0\}}$. For $k \in \Natural$ and $j=1, \ldots, n^k$, define $\Delta^k_j S := S_{\tau^k_j\wedge \cdot} - S_{\tau^k_{j-1}\wedge \cdot}$. Further, $S^\co$ is the continuous local martingale part of the semimartingale $S$. Since $S - S_0 = (S_- \indic_{\{S_- > 0 \}}) \cdot R$, we can write \eqref{eq: log-relative mult 1} as
\begin{eqnarray} \label{eq: log-relative mult 2}
  \log \pare{\frac{X^{(x, \pi; \parti^k)}}{X^{(x, \pi)}}} &=& \sum_{j=1}^{n^k} \log \left(  1 + \eta_{\tau^k_{j-1} +} \Delta^k_j S \right) \\
\nonumber    &-& \left( \eta \cdot S - \frac{1}{2} [\eta \cdot S^\co, \eta \cdot S^\co] -  \sum_{t \leq \cdot} \left( \eta_t \Delta S_t - \log \left( 1 + \eta_t \Delta S_t \right) \right) \right).
\end{eqnarray}
As $(\parti^k)_{k \in \Natural}$ converges to the identity and $\eta$ is \caglad, the dominated convergence theorem for stochastic integrals 
gives $\ucp \text{-} \lim_{k \to \infty} \sum_{j=1}^{n^k} \eta_{\tau^k_{j-1} +} \Delta^k_j S = \eta \cdot S$.
Furthermore, using the fact that the function $\Real \ni x \mapsto x - \log(1 + x)$ behaves like $\Real \ni x \mapsto x^2 / 2$ near $x = 0$, one obtains
\[
\ucp \text{-} \lim_{k \to \infty} \sum_{j=1}^{n^k} \left( \eta_{\tau^k_{j-1} +} \Delta^k_j S - \log \left( 1 + \eta_{\tau^k_{j-1} +} \Delta^k_j S \right) \right) = \sum_{t \leq \cdot} \left( \eta_t \Delta S_t - \log \left( 1 + \eta_t \Delta S_t \right) \right) + \frac{1}{2} [\eta \cdot S^\co, \eta \cdot S^\co].
\]
via standard stochastic-analysis manipulation. The last facts, coupled with \eqref{eq: log-relative mult 2}, readily imply that $\ucp$-$\lim_{k \to \infty} \log X^{(x, \pi; \parti^k)} = \log X^{(x, \pi)}$, which completes the proof.
\end{proof}

\subsection{Proof of Theorem \ref{thm: approx_with_bandh}} \label{subsec: proof of thm: approx with bandh}

Consider $X \equiv X^{(x, \pi)} \in \oX(x)$ for some $\simplex$-valued predictable process $\pi$. In order to prove Theorem \ref{thm: approx_with_bandh}, we can safely assume that $X \geq \epsilon$ for some $\epsilon > 0$, since if $X \in \oX(x)$, then $\epsilon + (1 - \epsilon/x) X \in \oX(x)$ as well. This assumption is in force throughout the proof.

Recall that a \textsl{simple} predictable process is of the form $\sum_{j=1}^n h_{j-1} \indic_{\dbraoc{t_{j-1}, t_j}}$, where $h_{j-1} \in \F_{t_{j-1}}$ for $j=1, \ldots, d$ and $0 = t_0 < t_1 < \ldots < t_n$, where $t_j \in \Real_+$ for $j=0, \ldots, n$. We shall show below that there exists a sequence of \emph{simple} $\simplex$-valued predictable processes $(\pi^k)_{k \in \Natural}$ such that $\ucp$-$\lim_{k \to \infty} X^{(x, \pi^k)} =  X^{(x, \pi)}$. Given the existence of such sequence, one can invoke Theorem \ref{thm: multipl_approx} and obtain a sequence $(X^k)_{k \in \Natural}$ of $\X(x)$-valued processes with $\ucp$-$\lim_{k \to \infty} X^k =  X$. 

To obtain the existence of a sequence of simple $\simplex$-valued predictable processes as described in the above paragraph, observe first that a use of the monotone class theorem provides the existence of a sequence $(\pi^k)_{k \in \Natural}$ of $\simplex$-valued, predictable, simple processes such that $\ucp$-$\lim_{k \to \infty} \pi^k \cdot R =  \pi \cdot R$, $\ucp$-$\lim_{k \to \infty} [(\pi^k -  \pi) \cdot R, \, (\pi^k -  \pi) \cdot R] = 0$, and $\inner{\pi^k}{\Delta R} > -1$ for all $k \in \Natural$.
Indeed, a simple approximation argument shows that only the special case when $\pi = v \indic_\Sigma$, with $v \in \simplex$ and $\Sigma$ is predictable and vanishes outside $\dbra{0, T}$ for some $T \in \Real_+$, has to be treated. Then, one uses the fact that the predictable $\sigma$-field on $\Omega \times \Real_+$ is generated by the algebra of simple predictable sets of the form $\bigcup_{j = 1}^n H_{j-1} \times (t_{j-1}, t_j]$, where $n \in \Natural$, $0 = t_0 < \ldots < t_n$ and $H_{j-1} \in \F_{t_{j-1}}$ for $j=1, \ldots, n$, and the claim readily follows.

Now, with $Y^k \dfn \pi^k \cdot R$ and $Y = \pi \cdot R$, the facts $\Exp(Y^k) > 0$ for all $k \in \Natural$ as well as $\Exp(Y) > 0$ allow one to write
\[
\log \left( \frac{\Exp(Y^k)}{\Exp(Y)} \right) = Y^k - Y - \frac{1}{2} \left( [Y^k, Y^k]^\co - [Y, Y]^\co\right) - \sum_{t \leq \cdot} \left( \Delta Y^k_t - \Delta Y_t - \log \left( \frac{1 + \Delta Y^k_t}{1 + \Delta Y_t} \right) \right).
\]
Using $\ucp$-$\lim_{k \to \infty} [Y^k - Y, Y^k - Y] =  0$ and $\ucp$-$\lim_{k \to \infty} Y^k =  Y$, which also imply that $\ucp$-$\lim_{k \to \infty} [Y^k, Y^k]^\co =  [Y, Y]^\co$ and $\ucp$-$\lim_{n \to \infty} \Delta Y^n =  \Delta Y$, we get $\ucp$-$\lim_{n \to \infty} \Exp( Y^n ) =  \Exp (Y)$, which is exactly what we wished to establish. \qed

\section{Application to the Expected Utility Maximization Problem} \label{sec: util max}

In this section we show that, for expected-utility-maximizing economic agents, allowing only simple trading with appropriately high trading frequency, results in indirect utilities and wealth processes that can be brought arbitrarily close to their theoretical continuous-trading \emph{optimal} counterparts.

%
%
%
\subsection{The utility maximization problem}

A \textsl{utility function} is an increasing and concave function $U : (0 , \infty) \mapsto \Real$. We also set $U(0) \dfn \downarrow \lim_{x \downarrow 0} U(x)$ to extend the definition of $U$ to cover zero wealth. Note that \emph{no} regularity conditions are hereby imposed on $U$.

In what follows, we fix a \emph{finite} stopping time $T$ that should be regarded as the financial planning horizon of an economic agent in the market. We then define the agent's indirect utility that can be achieved when continuous-time trading is allowed via
\begin{equation} \label{prob: util-max}
u (x) \, := \, \sup_{X \in \oX (x)} \expec \big[ U(X_T) \big].
\end{equation}
Observe that $\ou$ is a concave function of $x \in \Real_+$ and that $\ou (x) < \infty$ for \emph{some} $x > 0$ if and only if $\ou  (x) < \infty$ for \emph{all} $x \in \Real_+$. In particular, if $\ou (x) < \infty$ for some $x > 0$, $u$ is a proper continuous concave function. If $U$ is \emph{strictly} concave (in which case it is \emph{a fortiori} strictly increasing as well) and a solution to the utility maximization problem defined above exists, then it is necessarily unique.

Similarly, define the agent's indirect utility under simple, no-short-sales trading via
\begin{equation} \label{prob: util-max-simple}
\us (x) \, := \, \sup_{X \in \X (x)} \expec \big[ U(X_T) \big].
\end{equation}
It is obvious that $\us \leq \ou$. All the above remarks concerning $\ou$ carry over to $\us$ \emph{mutatis-mutandis}. Observe however that in almost no case is the supremum in \eqref{prob: util-max-simple} achieved. In other words, it is extremely rare that  an optimal wealth process in the class of simple trading strategies exists for the given utility maximization problem.

\subsection{Near-optimality using simple strategies}

We now show that the value functions $\us$ and $\ou$ are actually equal and that ``near optimal'' wealth processes under simple trading approximate arbitrarily close the solution of the continuous trading case, if the latter exists.

\begin{thm} \label{thm: util_max}
In what follows, condition \naonee \ of Definition \ref{dfn: NUPBR} is assumed. Using the notation introduced above, the following hold:
\begin{enumerate}
  \item $\us (x) = \ou (x)$ for all $x \in \Real_+$.
  \item Suppose that $U$ is \emph{strictly} concave and that $u < \infty$. Then, for any $x \in \Real_+$, any $\X(x)$-valued sequence $(X^k)_{k \in \Natural}$ and any $\oX(x)$-valued sequence $(\ox^k)_{k \in \Natural}$ with $\lim_{k \to \infty} \expec[U(X^k_T)] = \ou(x) = \lim_{k \to \infty} \expec[U(\ox^k_T)]$, we have $\prob$-$\lim_{k \to \infty} |X^k_T - \ox^k_T| = 0$.
  \item Suppose that $U$ is strictly concave and continuously differentiable, and that for some $x \in \Real_+$ there exists $\ox \in \oX(x)$ with $\ox > 0$, $\expec[U(\ox_T)] = \ou(x) < \infty$, and $\expec [U'(\ox_T) X_T] \leq \expec [U'(\ox_T) \ox_T] < \infty$ holding for all $X \in \oX(x)$. Then, for any $\X(x)$-valued sequence $(X^k)_{k \in \Natural}$ with $\lim_{k \to \infty} \expec[U(X^k_T)] = \ou(x)$, we have $\prob$-$\lim_{k \to \infty} \sup_{t \in [0, T]} |X_t^k - \ox_t| = 0$.
\end{enumerate}
\end{thm}

The proof of Theorem \ref{thm: util_max} is given in \S \ref{subsec: proof of thm: util max}; in \S \ref{subsec: result on supermart conv}, an interesting intermediate result is stated and proved. However, we first list a few remarks on the  assumptions and statements of Theorem \ref{thm: util_max}.

\begin{rem}
Under the mild assumption that \emph{for all $i=1, \ldots, d$, we have $S^i_{\zeta^i -} > 0$ on  $\{ \zeta^{S^i} < \infty\}$} (the asset lifetimes $\zeta^{S^i}$ were introduced in the statement of Theorem \ref{thm: semimarts}), condition \naone \ needed in the statement of Theorem \ref{thm: util_max} is actually equivalent to the semimartingale property of $S$. For more information, check Theorem 2.3 in \cite{KarPla08a}.
\end{rem}

\begin{rem}
The utility maximization problem for continuous trading has attracted a lot of attention and has been successfully solved using convex duality methods. In particular, in \cite{MR1722287} and \cite{MR2023886} it is shown that an optimal solution (wealth process) to problem \eqref{prob: util-max} exists for \emph{all} $x \in \Real_+$ and fixed financial planning horizon $T$ under the following conditions: $U$ is strictly concave and continuously differentiable in $(0, \infty)$, satisfies the \textsl{Inada} conditions $\lim_{x \downarrow 0} U'(x) = + \infty$, $\lim_{x \uparrow + \infty} U'(x) = 0$, as well as a \emph{finite dual value function} condition. These conditions can be used to ensure existence of the optimal wealth process in statement (3) of Theorem \ref{thm: util_max}, that additionally satisfies the prescribed properties mentioned there.
\end{rem}

\begin{rem}
In statements (2) and (3), strict concavity of $U$ cannot be dispensed with in order to obtain the result: even in cases where the supremum in \eqref{prob: util-max} is attained, the absence of strict concavity implies that the optimum is not necessarily unique.
\end{rem}

\begin{rem}
Even if we not directly assume condition \naone  \ in statement (3), it is indirectly in force because of the existence of $\ox \in \oX(x)$ with $\ox > 0$ and $\expec[U(\ox)] = \ou(x) < \infty$. Indeed, suppose that \naone \ fails and pick $T \in \Real_+$ and $(X^n)_{n \in \Natural}$ such that $X^n \in \X (1/n)$ and $\prob[X^n_T \geq \xi] =1$ for all $n \in \Natural$, where $\prob[\xi \geq 0] = 1$ and $\prob[\xi > 0] > 0$. In that case, the convexity of $\oX(x + 1/n)$  gives that $(\ox + X^n) \in \oX(x + 1/n)$ for all $n \in \Natural$. Therefore, 
\[
\ou (x + 1/n) \geq \expec \bra{U(\ox + X^n)} \geq \expec \bra{U(\ox +  \xi)} > \ou (x)
\]
holds for all $n \in \Natural$, which implies that $\ou (x) < \lim_{n \to \infty} \ou(x + 1/n)$ and contradicts the continuity of the finitely-valued function $\ou$. For a similar result in this direction, see Proposition 4.19 in \cite{MR2335830}.
\end{rem}

\begin{rem}
The difference between statements (2) and (3) in Theorem \ref{thm: util_max} is that in the latter case we can infer \emph{uniform} convergence of the wealth processes to the limiting one, while in the former we only have convergence of the terminal wealths. It is an open question whether the uniform convergence of the wealth processes can be established without assuming that the utility maximization problem involving continuous trading has a solution.
\end{rem}

%

\subsection{An result on supermatingale convergence} \label{subsec: result on supermart conv}

The following result, stated separately due to its independent interest, will help proving statement (3) of Theorem \ref{thm: util_max}.

\begin{prop} \label{prop: supermart conv}
On the filtered probability space $(\Omega, \, (\F_t)_{t \in \Real_+}, \, \qprob)$, let $(Z^k)_{k \in \Natural}$ be a sequence of nonnegative $\qprob$-supermartingales on $\dbra{0, T}$ with $Z^k_0 = 1$ for all $k \in \Natural$ and $\qprob$-$\lim_{k \to \infty} Z^k_T = 1$, where $T$ is a finite stopping time. Then, $\qprob$-$\lim_{k \to \infty} \sup_{t \in [0, T]} |Z_t^k - 1| = 0$.
\end{prop}

\begin{proof}

Since $Z^k \geq 0$ for all $k \in \Natural$, it suffices to show that $\qprob$-$\lim_{k \to \infty} \sup_{t \in [0, T]} Z_t^k = 1$ and $\qprob$-$\lim_{k \to \infty} \inf_{t \in [0, T]} Z^k_t = 1$.

For proving $\qprob$-$\lim_{k \to \infty} \sup_{t \in [0, T]} Z_t^k = 1$, observe that $\lim_{k \to \infty} \expec^\qprob[Z^k_T] = 1$ as a consequence of Fatou's lemma; this implies the $\qprob$-uniform integrability of $(Z^k_T)_{k \in \Natural}$ and as a consequence we obtain $\lim_{k \to \infty} \expec^\qprob[|Z^k_T - 1 |] = 0$. In particular, the probabilities $(\qprob^k)_{k \in \Natural}$ defined on $(\Omega, \F_T)$ via $(\ud \qprob^k / \ud \qprob)|_{\F_T} = Z^k_T / \expec^\qprob [Z^k_T]$ converge in total-variation norm to $\qprob$.

Fix $\epsilon > 0$ and let  $\tau^k :=  \inf \{t \in \Real_+ \such Z^k_t > 1 + \epsilon \} \wedge T$. We have $\expec^\qprob[Z^k_T] \leq \expec^\qprob[Z^k_{\tau^k}] \leq 1$, which means that $\lim_{k \to \infty} \expec^\qprob[Z^k_{\tau^k}] = 1$. Showing that $\lim_{k \to \infty} \qprob[\tau^k < T] = 0$ will imply that $\qprob$-$\lim_{k \to \infty} \sup_{t \in [0, T]} Z^k_t = 1$, since $\epsilon > 0$ is arbitrary. Suppose on the contrary (passing to a subsequence if necessary) that $\lim_{k \to \infty} \qprob[\tau^k < T] = p > 0$.  Then,
\begin{eqnarray*}
  1 = \lim_{k \to \infty} \expec^\qprob[Z^k_{\tau^k}] &\geq& (1 + \epsilon) p + \liminf_{k \to \infty} \expec^\qprob [Z^k_{T} \indic_{\{ \tau^k = T \}}] \\
    &=& (1 + \epsilon) p + \liminf_{k \to \infty} \left( \expec^\qprob[Z^k_T] \qprob^k [ \tau^k = T ] \right) \ = \ 1 + \epsilon p,
\end{eqnarray*}
where the last equality follows from $\lim_{k \to \infty} \expec^\qprob[Z^k_T] = 1$ and $\lim_{k \to \infty} \qprob^k [ \tau^k = T ] = \lim_{k \to \infty} \qprob [ \tau^k = T ] = 1 - p$. This contradicts $p > 0$ and the first claim is proved.

Again, with fixed $\epsilon > 0$, redefine $\tau^k :=  \inf \{t \in [0, T] \such Z^k_t < 1 - \epsilon \} \wedge T$ --- we only need to show that $\lim_{k \to \infty} \qprob[\tau^k < T] = 0$. Since $\qprob[Z^k_T > 1 - \epsilon^2 \such \F_{\tau^k}] \leq (1 - \epsilon) / (1 - \epsilon^2) = 1 / (1 + \epsilon)$ holds on the event $\{ \tau^k < T \}$, we have
\[
\qprob[Z^k_T > 1 - \epsilon^2] = \expec^\qprob \big[ \qprob[Z^k_T > 1 - \epsilon^2 \such \F_{\tau^k}] \big] \ \leq \ \qprob[\tau^k = T] + \qprob[\tau^k < T] \frac{1}{1+\epsilon}.
\]
Use $\qprob[\tau^k = T] = 1 - \qprob[\tau^k < T]$, rearrange and take the limit as $n$ goes to infinity to obtain
\[
\limsup_{k \to \infty} \qprob[\tau^k < T] \ \leq \ \frac{1 + \epsilon}{\epsilon} \limsup_{k \to \infty} \qprob[Z_T^k \leq 1 - \epsilon^2] \ = \ 0,
\]
which completes the proof of Proposition \ref{prop: supermart conv}.
\end{proof}

\subsection{Proof of Theorem \ref{thm: util_max}} \label{subsec: proof of thm: util max}

We close by giving the proof of each of the three statements of Theorem \ref{thm: util_max}.

\subsubsection{Proof of statement (1)} \label{subsubsec: proof ot util max 1}
We begin by proving that $\us = \ou$. Assume first that $\ou$ is finite. Since $\lim_{\epsilon \downarrow 0} \ou(x - \epsilon) = \ou(x)$ for all $x > 0$, it suffices to prove that for all $\epsilon \in (0, x)$ there exists an $\X(x)$-valued sequence $(X^k)_{k \in \Natural}$ such that $\ou(x - \epsilon) \leq \liminf_{k \to \infty} \expec[U(X^k_T)] + \epsilon$. Pick $\xi \in \oX(x - \epsilon)$ such that $\expec[U(\xi_T)] \geq \ou(x - \epsilon) - \epsilon$; then, $X := \epsilon + \xi$ satisfies $\expec[U(X_T)] \geq \ou(x - \epsilon) - \epsilon$, $X \in \oX(x)$ and $X \geq \epsilon$. According to Theorem \ref{thm: approx_with_bandh} (combined with Remark \ref{rem: approx away from zero}), we can find an $\X(x)$-valued sequence $(X^k)_{k \in \Natural}$ with $\prob$-$\lim_{k \to \infty} X^k_T = X_T$ and $X_T^k \geq \epsilon$. Fatou's lemma implies that $\expec[U(X_T)] \leq \liminf_{k \to \infty} \expec[U(X^k_T)]$ and the proof that $\us = \ou$ for the case of finitely-valued $\ou$ is clarified. The case where $\ou \equiv \infty$ is treated similarly.

\subsubsection{Proof of statement (2)} \label{subsubsec: proof ot util max 2}

Pick any $\X(x)$-valued sequence $(X^k)_{k \in \Natural}$ and any $\oX(x)$-valued sequence $(\ox^k)_{k \in \Natural}$ with $\lim_{k \to \infty} \expec[U(X^k_T)] = \ou(x) = \lim_{k \to \infty} \expec[U(\ox^k_T)]$. We shall show below that $\prob$-$\lim_{k \to \infty} |X^k_T - \ox^k_T| = 0$.

For any $m \in \Natural$, define $K_m := \{ (a, b) \in \Real^2 \ | \  a \in [0, m], b \in [0, m] \textrm{ and } |a - b| > 1 / m \}$. As follows from Proposition 2.1 of \cite{KarPla08a}, under condition \naone \ both sequences $(X^k_T)_{k \in \Natural}$ and $(\ox^k_T)_{k \in \Natural}$ are bounded in probability. Therefore, $\prob$-$\lim_{k \to \infty} |X_T^k - \ox^k_T| = 0$ will follow if we establish that, for all $m \in \Natural$, $\lim_{k \to \infty} \prob \big[ \big( X_T^k, \, \ox^k_T \big) \in K_m \big] = 0$.

Fix some $m \in \Natural$; the strict concavity of $U$ implies the existence of some $\beta_m > 0$ such that for all $(a, b) \in (0, \infty) \times (0, \infty)$ we have
\[
\frac{U(a) + U(b)}{2} + \beta_m
\indic_{K_m} (a, b) \leq U \Big( \frac{a + b}{2} \Big).
\]
Setting $a = X^k_T$, $b = \ox^k_T$ in the previous inequality and taking expectations, one gets
\begin{eqnarray*}
  \beta_m \prob \big[ \big( X_T^k, \, \ox^k_T \big)
\in K_m \big] &\leq& \expec \left[ U \left(\frac{X_T^k + \ox^k_T}{2} \right) \right] - \frac{\expec[U(X^k_T)] + \expec[U(\ox^k_T)]}{2} \\ &\leq& \ou(x) - \frac{\expec[U(X^k_T)] + \expec[U(\ox^k_T)]}{2};
\end{eqnarray*}
since $\lim_{k \to \infty} \big( \expec[U(X^k_T)] + \expec[U(\ox^k_T)] \big) = 2 \ou(x)$, $\lim_{k \to \infty} \prob \big[ \big( X_T^k, \, \ox^k_T \big) \in K_m \big] = 0$ follows.

\subsubsection{Proof of statement (3)} \label{subsubsec: proof ot util max 3}

Pick any $\X(x)$-valued sequence $(X^k)_{k \in \Natural}$ with the property that $\lim_{k \to \infty} \expec[U(X^k_T)] = \ou(x)$. We already know from part (2) of Theorem \ref{thm: util_max} that $\prob$-$\lim_{k \to \infty} X_T^k = \ox_T$. What remains in order to prove statement (3) is to pass to the stronger convergence $\ucp$-$\lim_{k \to \infty} X^k = \ox$. Observe that since $\inf_{t \in [0, T]} \ox_t > 0$, which is a consequence of $\ox > 0$ and condition \naone, the latter convergence is equivalent to $\ucp$-$\lim_{k \to \infty} (X^k / \ox) = 1$.

Define a new probability $\qprob$ on $\F_\infty$ via the recipe
\[
\frac{\ud \qprob}{\ud \prob} = \frac{\ox_T U'(\ox_T)}{\expec [\ox_T U'(\ox_T)]}.
\]
The assumptions of statement (3) in Theorem \ref{thm: util_max} imply that $\qprob$ is well-defined and equivalent to $\prob$ on $\F_\infty$, as well as that $X / \ox$ is a $\qprob$-supermartingale on $\dbra{0, T}$ for all $X \in \oX(x)$.
Letting $Z^k := X^k / \ox$ for all $k \in \Natural$, we are in the following situation: $Z^k$ is a nonnegative $\qprob$-supermartingale on $\dbra{0, T}$ with $Z^k_0 = 1$ for all $k \in \Natural$, and $\qprob$-$\lim_{k \to \infty} Z^k_T = 1$. Then, Proposition \ref{prop: supermart conv} allows us to conclude.

\bibliographystyle{siam}
\bibliography{buy_and_hold}
\end{document}